\documentclass[conference, 10 pt]{ieeeconf}
\usepackage{times}
\usepackage[utf8]{inputenc}

\usepackage{amsmath}
\usepackage{amsfonts}
\usepackage{amssymb}
\usepackage{mathtools}
\usepackage{cite}
\usepackage[font=footnotesize]{caption}
\usepackage{multicol}
\usepackage[hidelinks]{hyperref}
\usepackage{algorithm}
\usepackage{amsthm}
\usepackage[]{graphicx}
\usepackage{subcaption}
\usepackage[table]{xcolor}
\usepackage{flushend}

\newtheorem{theorem}{Theorem}
\newtheorem{lemma}{Lemma}

\definecolor{lightgray}{gray}{0.9}

\newcommand{\figsize}{0.35\textwidth}
\newtheorem{problem}{Problem}

\IEEEoverridecommandlockouts 
\author{Diogo Almeida\textsuperscript{1} and Yiannis Karayiannidis\textsuperscript{2}
\thanks{\textsuperscript{1}Division of Robotics, Perception and Learning, KTH Royal Institute of Technology, SE-100 44 Stockholm, Sweden {\tt{diogoa@kth.se}} \newline
\indent\textsuperscript{2}Dept. of Electrical Eng., Chalmers University of Technology, SE-412 96 Gothenburg, Sweden {\tt{yiannis}@chalmers.se}\newline
This work is partially supported by the Swedish Foundation for
Strategic Research project GMT14-0082 FACT and by UNIFICATION, Vinnova, Produktion 2030.}
}

\title{A Lyapunov-Based Approach to Exploit Asymmetries in Robotic Dual-Arm Task Resolution}

\begin{document}
%\mathtoolsset{showonlyrefs}
\maketitle
\begin{abstract}
	Dual-arm manipulation tasks can be prescribed to a robotic system in terms of desired absolute and relative motion of the robot's end-effectors.
	These can represent, e.g., jointly carrying a rigid object or performing an assembly task.
	When both types of motion are to be executed concurrently, the symmetric distribution of the relative motion between arms prevents task conflicts.
	Conversely, an asymmetric solution to the relative motion task will result in conflicts with the absolute task.
	In this work, we address the problem of designing a control law for the absolute motion task together with updating the distribution of the relative task among arms.
	Through a set of numerical results, we contrast our approach with the classical symmetric distribution of the relative motion task to illustrate the advantages of our method.
\end{abstract}
\IEEEpeerreviewmaketitle

\section{Introduction}
Dual-armed robots are able to overcome payload or dexterity limits of a single arm. 
These limits are addressed by leveraging the potential of cooperative motions of the two arms in the robot system, as two arms are able to jointly carry an heavy object, or perform complex assembly tasks independently of environmental fixtures, to name two examples.

Dual-armed robotic manipulation tasks can be broadly divided into two categories: \textit{relative} and \textit{absolute} tasks.
Relative tasks can be described solely through a relative motion of the robot's end-effectors. 
These include challenges such as the assembly of two components \cite{Ajoudani2014, Almeida2016, Almeida2016b, Stavridis2018}, drawing \cite{Lee2013, Lee2015} or machining\cite{Lee2012, Owen2005, Owen2008}.
Absolute tasks, conversely, consist in problems that can be solved by assuming a rigid connection between the robot end-effectors. 
Some examples of absolute motion tasks are the cooperative manipulation of rigid objects, tools and mechanisms in the robot's environment \cite{Caccavale2008, Heck2013, Wang2015, Erhart2015, Ajoudani2014b}.

Initial work on dual-arm manipulation emphasized solutions based on master-slave approaches \cite{Uchiyama1987}. 
Relative motion tasks can be solved by adopting such formulations, where one arm, the master, is tasked with executing the gross motion required to solve the problem, while the slave arm adopts a passive role, such as complying to contact forces.
Absolute motion tasks can be realized by requiring the slave arm's end-effector to keep a constant relative pose w.r.t the master.

\begin{figure}
	\centering
	\begin{subfigure}{\figsize}
		\includegraphics[width = \textwidth]{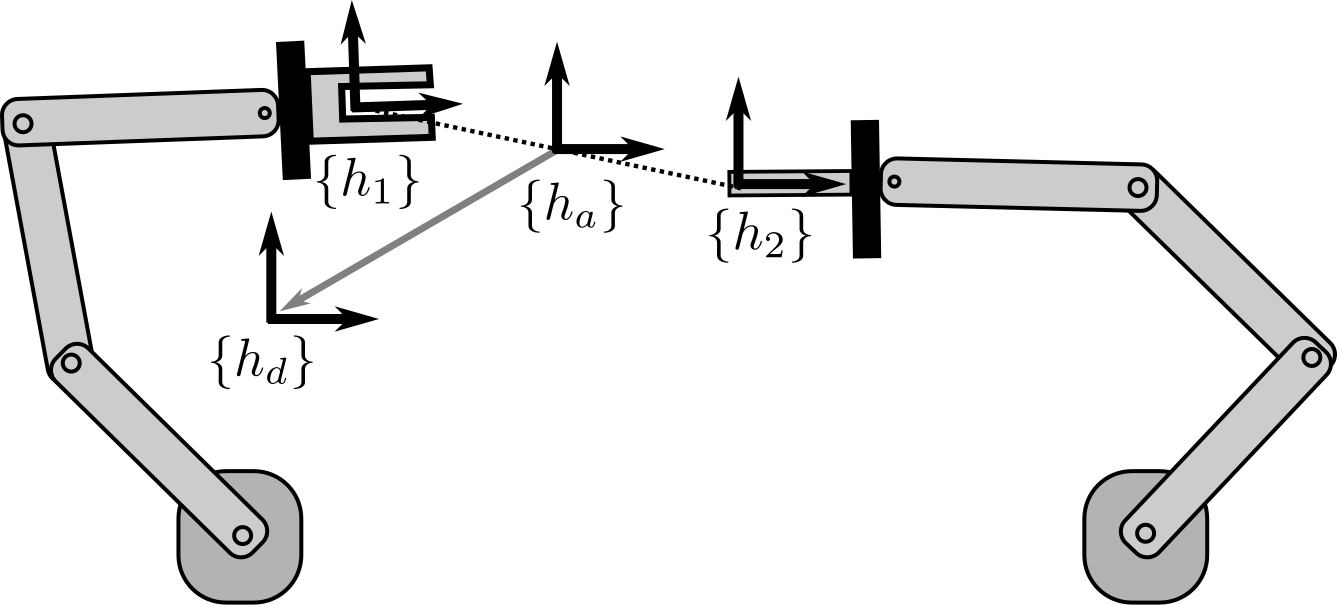}
	\end{subfigure}
	\caption{Illustration of a robotic dual-arm cooperative task. If the relative velocity of the system is defined along the dotted line and the absolute motion is aligned with the gray arrow, an \textit{asymmetric} distribution of the relative velocity between the manipulators will disturb the absolute motion task. \label{diagram}}
\end{figure}
Alternatively to a master-slave approach, the Cooperative Task Space (CTS) \cite{Chiacchio1996} formulation specifically defines an absolute and a relative motion space along which the two types of tasks can be specified.
This choice is based on earlier cooperative solutions, which rely on the assumption that a rigid object is being jointly held by the dual-armed system \cite{Uchiyama1988, Chiacchio1989, Chiacchio1992}. This results in a \textit{symmetric} solution to the relative motion task, i.e., each arm contributes equally to the relative motion.
However, in CTS, the assumption of grasping a rigid object is no longer a requirement to derive the necessary task-space relationships.

Recent work has proposed an Extended CTS (ECTS) \cite{Park2015}, which enables an asymmetric execution of the relative task, i.e., the user can specify, through a parameter, the degree to which each arm contributes to the relative motion task.
This degree of cooperation is a convenient expression of the redundancy of the relative task: a master-slave solution can be obtained by assigning the relative task entirely to one of the arms, while symmetric behavior results from an equal distribution of the relative task among arms. 
This can be exploited, e.g., to optimize the performance of the system when one of the arms is in a less dexterous configuration~\cite{Park2016}.

In this work, we design a control solution to the absolute motion task which leverages induced asymmetries on the execution of the relative task.
We note that it is not possible to asymmetrically execute a desired relative motion without affecting the absolute task of the system.
This observation can be used as the basis to jointly design a controller for the absolute motion task and an update law for the parameter which governs the arms' cooperation in the execution of the relative motion task.
This is a novel approach to the problem of the cooperative control of dual-armed robotic systems.
We present a case study focusing on the linear component of the problem, section \ref{case_study}, and address the full dual-arm manipulation problem in section \ref{dual_arm}.
Throughout our article, we compare our approach with the classical CTS formulation in a sequence of numerical examples.

\section{Preliminaries}
Consider a dual-armed robotic system, composed by two robotic manipulators.
Each manipulator consists of a set of links, connected through $n_i \in \mathbb{N}$ generalized (i.e., revolute or prismatic) joints $\mathbf{q}_i \in \mathbb{R}^{n_i}$, where $i \in \{1, 2\}$ identifies the manipulator.
The pose of each manipulators' end-effector can be represented by a frame $\{h_i\}$, depicted in Fig. \ref{diagram}, which contains translational $\mathbf{p}_i \in \mathbb{R}^3$ and rotational $\mathbf{R}_i \in SO(3)$ components. 
Let $\mathbf{v}_i \in \mathbb{R}^6$ denote the cartesian twist of the $i$-th manipulator. 
The relationship between cartesian twists and joint velocities $\dot{\mathbf{q}}_i$ is given by 
\begin{equation}
	\mathbf{v}_i = \mathbf{J}_i(\mathbf{q}_i) \dot{\mathbf{q}}_i,
\end{equation}
where $\mathbf{J}_i(\mathbf{q}_i) \in \mathbb{R}^{6 \times n_i}$ is the Jacobian matrix of the $i$-th manipulator and $\mathbf{v}_i = [\dot{\mathbf{p}}_i^\top, \boldsymbol{\omega}_i^\top]^\top$, where $\dot{\mathbf{p}}_i, \boldsymbol{\omega}_i \in \mathbb{R}^3$ are, respectively, the linear and angular velocity parts of the twist.
For the following, we will assume that each manipulator has $n_i \geq 6$, such that each arm has at least as many degrees-of-freedom as its task space.

\subsection{Cooperative Task Space}
The CTS formulation defines two motion frames, the \textit{absolute} motion frame, $\{h_a\}$, and the \textit{relative} motion frame, $\{h_r\}$, such that a desired relative motion of the end-effectors can be specified in $\{h_r\}$ and, conversely, an absolute (i.e., common) motion can be prescribed in $\{h_a\}$. 
The position and orientation of $\{h_a\}$ and $\{h_r\}$ are given by
	\begin{align}
		&\mathbf{p}_a = \frac{1}{2}(\mathbf{p}_1 + \mathbf{p}_2)&  &\mathbf{R}_a = \mathbf{R}_1 \mathbf{R}_{\mathbf{k}_{1, 2}}\left (\frac{\vartheta_{1, 2}}{2} \right) \label{absolute_frame}\\
		&\mathbf{p}_r = \mathbf{p}_2 - \mathbf{p}_1& &\mathbf{R}_r = \mathbf{R}_1^\top \mathbf{R}_2, \label{relative_frame}
	\end{align}
where $\mathbf{R}_{\mathbf{k}_{1, 2}}\left (\vartheta_{1, 2} \right)$ denotes the angle-axis representation of $\mathbf{R}_r$, with $\vartheta_{1, 2} \in \mathbb{R}$ being the angle one needs to rotate $\mathbf{R}_2$ about the axis $\mathbf{k}_{1, 2} \in \mathbb{R}^3$ to obtain $\mathbf{R}_1$.
Note that the time derivatives of \eqref{absolute_frame} and \eqref{relative_frame} lead to the following relationship between task-space twists \cite{Chiacchio1996},
\begin{equation}
	\begin{bmatrix}
		\mathbf{v}_a\\
		\mathbf{v}_r
	\end{bmatrix} = \begin{bmatrix}
		\frac{1}{2}\mathbf{I}_6 & \frac{1}{2} \mathbf{I}_6\\
		-\mathbf{I}_6 & \mathbf{I}_6
	\end{bmatrix} \begin{bmatrix}
		\mathbf{v}_1\\
		\mathbf{v}_2
	\end{bmatrix},
	\label{cts_relationship}
\end{equation}
where $\mathbf{I}_n$ denotes the $n$-dimensional identity matrix and $\mathbf{v}_a = [\dot{\mathbf{p}}_a^\top, \boldsymbol{\omega}_a^\top]^\top$ and $\mathbf{v}_r = [\dot{\mathbf{p}}_r^\top, \boldsymbol{\omega}_r^\top]^\top$ are, respectively, the absolute and relative motion twists in the CTS formulation.
The inverse relationship is straightforward to obtain,
\begin{equation}
	\begin{bmatrix}
		\mathbf{v}_1\\
		\mathbf{v}_2
	\end{bmatrix} = \begin{bmatrix}
		\mathbf{I}_6 &  -\frac{1}{2}\mathbf{I}_6\\
		\mathbf{I}_6 & \frac{1}{2} \mathbf{I}_6
	\end{bmatrix} \begin{bmatrix}
		\mathbf{v}_{a}\\
		\mathbf{v}_{r}
	\end{bmatrix}.
	\label{cts_inv}
\end{equation}
In CTS, the user specifies desired absolute and relative velocities, respectively $\mathbf{v}_{a_d}$ and $\mathbf{v}_{r_d}$, which are resolved to end-effector velocities through \eqref{cts_inv}.
This is in contrast to master-slave approaches, which assign the execution of the task to one end-effector, with the other adopting a passive stance, such as regulating contact forces.
The CTS fomulation results in a \textit{symmetric} solution of the dual-arm relative motion.
Both end-effectors will contribute equally to the relative motion, as $\mathbf{v}_{r}$ is distributed in the same proportion among the arms.

\subsection{Asymmetric relative motion}
In addition to the master-slave or symmetric solution to the relative motion problem, we can consider \textit{blended}, or asymmetric, modes of cooperation.
The ECTS formulation in particular redefines $\{h_a\}$ as a weighted version of \eqref{absolute_frame}, $\{h_{a_E}\}$, such that
\begin{equation}
	\begin{aligned}
		\mathbf{p}_{a_E} &= \alpha\mathbf{p}_1 + (1 - \alpha)\mathbf{p}_2\\
		\mathbf{R}_{a_E} &= \mathbf{R}_1 \mathbf{R}_{\mathbf{k}_{1, 2}}\left((1 - \alpha)\vartheta_{1,2} \right),
	\end{aligned}
	\label{extended_absolute_frame}
\end{equation}
where $\alpha \in D_\alpha  \triangleq \{x \in \mathbb{R} : 0 \leq x \leq 1\}$.
Under the new definition \eqref{extended_absolute_frame}, the relationship between task-space velocities is given by
\begin{equation}
	\begin{bmatrix}
		\mathbf{v}_{a_E}\\
		\mathbf{v}_r
	\end{bmatrix} = \begin{bmatrix}
		\alpha \mathbf{I}_6 & (1 - \alpha) \mathbf{I}_6\\
		-\mathbf{I}_6 & \mathbf{I}_6
	\end{bmatrix} \begin{bmatrix}
		\mathbf{v}_1\\
		\mathbf{v}_2
	\end{bmatrix},
\end{equation}
and the inverse relationship illustrates the effect of the cooperation parameter $\alpha$ on the resolution of the system's relative motion,
\begin{equation}
	\begin{bmatrix}
		\mathbf{v}_{1}\\
		\mathbf{v}_{2}
	\end{bmatrix} = \begin{bmatrix}
		\mathbf{I}_6 & -(1 - \alpha)\mathbf{I}_6\\
		\mathbf{I}_6 & \alpha \mathbf{I}_6
	\end{bmatrix} \begin{bmatrix}
		\mathbf{v}_{a_E}\\
		\mathbf{v}_{r}
	\end{bmatrix}.
	\label{ects_solution}
\end{equation}
In particular, $\alpha = 1$ or $\alpha = 0$ denotes a master-slave approach to the relative motion task, with different end-effectors adopting the role of master, while $\alpha = 0.5$ results in the symmetric approach. 
Other values of $\alpha$ affect the degree to which each arm cooperates on the resolution of the relative motion task and compose the blended mode of cooperation.

\section{Problem description \label{problem_statement}}
The ECTS solution \eqref{ects_solution} entails that, for $\alpha \neq 0.5$, $\mathbf{v}_{r}$ will act as a disturbance to the symmetric absolute motion, described in $\{h_a\}$, eq.~\eqref{absolute_frame}. 
Indeed, we have
\begin{equation}
	\mathbf{v}_a = \frac{1}{2} (\mathbf{v}_1 + \mathbf{v}_2) = (\alpha - 0.5) \mathbf{v}_{r} + \mathbf{v}_{a_E}.
	\label{vr_disturbance}
\end{equation}
This highlights the conflict between the asymmetric resolution of a relative motion task and the control of the symmetric absolute motion of the cooperative system.
Consider as an example a robot equipped with pruning shears, tasked with pruning plants on a garden. 
If the relative motion (i.e., operating the shears) is solved asymmetrically, this will affect the absolute task (i.e., moving the shears on the robot's workspace).

\subsection{Decomposing the linear and angular motion}
Note that the linear and angular terms of the relative twist $\mathbf{v}_{r}$ only impact the respective linear and angular components of $\mathbf{v}_a$ in \eqref{vr_disturbance}. 
Therefore, we will define different cooperation parameters, $\alpha_p, \alpha_{\omega} \in \mathbb{R}$ to, respectively, set the arms' cooperation degree on the translational and rotational tasks.
Let
\begin{equation}
	\boldsymbol{\Lambda}(\alpha_p, \alpha_\omega) \triangleq \begin{bmatrix}
		\alpha_p \mathbf{I}_3 & \mathbf{0}\\
		\mathbf{0} & \alpha_\omega \mathbf{I}_3
	\end{bmatrix}.
\end{equation}
In the remaining of this text, we will assume that an assigned cooperative task $\mathbf{v}_{a_d}, \mathbf{v}_{r_d}$ is resolved into desired end-effector twists $\mathbf{v}_{1_d}, \mathbf{v}_{2_d}$ as
\begin{equation}
	\begin{bmatrix}
		\mathbf{v}_{1_d}\\
		\mathbf{v}_{2_d}
	\end{bmatrix} = \begin{bmatrix}
		\mathbf{I}_6 & -(\mathbf{I}_6 - \boldsymbol{\Lambda}(\alpha_p, \alpha_\omega) )\\
		\mathbf{I}_6 & \boldsymbol{\Lambda}(\alpha_p, \alpha_\omega) 
	\end{bmatrix} \begin{bmatrix}
		\mathbf{v}_{a_{d}}\\
		\mathbf{v}_{r_d}
	\end{bmatrix}.
	\label{two_alphas_res}
\end{equation}
%\subsection{Problem statement}
Now, consider the problem of regulating the absolute frame \eqref{absolute_frame} of the cooperative system by assigning desired twists, $\mathbf{v}_{i_d}$, to each frame $\{h_i\}$, $i \in \{1, 2\}$,
\begin{problem}[Cooperative absolute motion]
	Let $\{h_{d}\}$ be the desired absolute motion frame for the cooperative system \eqref{absolute_frame}-\eqref{relative_frame}, which we will assume to be a constant reference, i.e., $\dot{\mathbf{p}}_d = \boldsymbol{\omega}_d = \mathbf{0}$.
	Additionally, consider the unit quaternion $\mathcal{Q}_a = \{\eta_a, \boldsymbol{\epsilon}_a\}$ as a representation of the absolute orientation $\mathbf{R}_a$, and the quaternion $\mathcal{Q}_d$ as the representation of a desired absolute orientation of the cooperative system. 	
	We define the absolute position error as $\tilde{\mathbf{p}}_a = \mathbf{p}_{d} - \mathbf{p}_a$, and the absolute orientation error as $\tilde{\mathcal{Q}}_a = \mathcal{Q}_d * \mathcal{Q}_a^{-1} = \{\tilde{\eta}_a, \tilde{\boldsymbol{\epsilon}}_a\}$. 
	Given a desired relative motion between the end-effectors, $\mathbf{v}_{r_d} = [\dot{\mathbf{p}}_{r_d}^\top, \boldsymbol{\omega}_{r_d}^\top]^\top$, we aim to prescribe desired velocities $\mathbf{v}_{1_d}$ and $\mathbf{v}_{2_d}$	such that $[\tilde{\mathbf{p}}_a^\top, \tilde{\boldsymbol{\epsilon}}_a^\top]^\top = \mathbf{0}$ is asymptotically stable.
	\label{cooperative_problem}
\end{problem}
Since our formulation \eqref{two_alphas_res} includes the symmetric solution, i.e., CTS, as the particular case of $\alpha_p = \alpha_\omega = 0.5$, we will only consider \eqref{two_alphas_res} in the remaining text.
Let $\mathbf{v}_{a_d} = [\dot{\mathbf{p}}_{a_d}^\top, \boldsymbol{\omega}_{a_d}^\top]^\top$ be the commanded absolute twist to the cooperative system \eqref{two_alphas_res}.
The dynamics of the absolute position error are given by
\begin{equation}
	\dot{\tilde{\mathbf{p}}}_a = -\dot{\mathbf{p}}_a = (0.5 - \alpha_p)\dot{\mathbf{p}}_{r_d} - \dot{\mathbf{p}}_{a_d},
	\label{pos_error_dyn}
\end{equation}
while the evolution of the absolute orientation error is given by the quaternion propagation equation \cite[p. 140]{Siciliano2008},
\begin{equation}
	\begin{aligned}
        \dot{\tilde{\eta}}_a &= -\frac{1}{2}\tilde{\boldsymbol{\epsilon}}_a^\top \boldsymbol{\omega}_{a}\\
        \dot{\tilde{\boldsymbol{\epsilon}}}_a &= \frac{1}{2}(\tilde{\eta}_a\mathbf{I}_3 - \mathbf{S}(\tilde{\boldsymbol{\epsilon}}_a))\boldsymbol{\omega}_{a},
    \end{aligned}
    \label{abs_ori_prop}
\end{equation}
where $\mathbf{S}(\boldsymbol{\epsilon})$ denotes the skew-symmetric matrix such that $\mathbf{S}(\boldsymbol{\epsilon})\boldsymbol{\omega} = \boldsymbol{\epsilon} \times \boldsymbol{\omega}$.
Note as well that $\boldsymbol{\omega}_a$ consists of the angular part of \eqref{vr_disturbance},
\begin{equation}
	\boldsymbol{\omega}_a = (\alpha_\omega - 0.5)\boldsymbol{\omega}_{r_d} + \boldsymbol{\omega}_{a_d}.
	\label{angular_vel}
\end{equation}
A possible solution to Problem \ref{cooperative_problem} is to set $\alpha_p = \alpha_\omega = 0.5$ and adopt the feedback control law \cite{Caccavale2000}
\begin{equation}
	\mathbf{v}_{a_d} = \begin{bmatrix}
		\mathbf{K}_p \tilde{\mathbf{p}}_a\\
		\mathbf{K}_\omega \tilde{\boldsymbol{\epsilon}}_a
	\end{bmatrix},
\end{equation}
with $\mathbf{K}_p, \mathbf{K}_\omega \in \mathbb{R}^{3\times 3}$ being positive definite.

Alternatively, we can assume that $\alpha_p$ and $\alpha_\omega$ are time-varying parameters, i.e., $\alpha_p \triangleq \alpha_p(t)$ and $\alpha_\omega \triangleq \alpha_\omega(t)$, and expand the state of \eqref{pos_error_dyn} and \eqref{abs_ori_prop} to include the cooperation parameters, with appropriately designed dynamics
\begin{equation}
	\dot{\alpha}_p = f_p(\tilde{\mathbf{p}}_a, \alpha_p), \qquad \dot{\alpha}_\omega = f_\omega(\tilde{\mathcal{Q}}_a, \alpha_\omega).
	\label{dynamic_alphas}
\end{equation}
For the joint system with dynamics given by \eqref{pos_error_dyn}, \eqref{abs_ori_prop} and \eqref{dynamic_alphas}, we define the desired equilibrium as
\begin{align}
	\tilde{\mathbf{p}}_a = \mathbf{0} & \quad\wedge \quad\alpha_p = 0.5 \label{pos_equi}\\
	\tilde{\mathcal{Q}}_a = \{1, \mathbf{0}\} & \quad\wedge \quad\alpha_\omega = 0.5 \label{ang_equi}.
\end{align}
In the following, we will assume that the relative velocity $\mathbf{v}_{r_d}$ is a known and bounded quantity, and we will show that by designing $f_p(\tilde{\mathbf{p}_a}, \alpha_p)$ and $f_\omega(\tilde{\mathcal{Q}}_a, \alpha_\omega)$ in conjunction with control laws for $\dot{\mathbf{p}}_{a_d}$ and $\boldsymbol{\omega}_{a_d}$, respectively, we are able to leverage an asymmetric execution to the relative motion task in the design of a solution to Problem \ref{cooperative_problem}.

\section{Case study: Coordinating two points \label{case_study}}
We will first consider the coordination of two $n$-dimensional points $\mathbf{p}_i \in \mathbb{R}^n$.
Each point is assumed to move holonomically.
In general, the dynamics of the two-point system are given by
\begin{equation}
	\begin{bmatrix}
		\dot{\mathbf{p}}_{1}\\
		\dot{\mathbf{p}}_{2}
	\end{bmatrix} = \begin{bmatrix}
		\mathbf{I}_n & -(1 - \alpha_p)\mathbf{I}_n\\
		\mathbf{I}_n & \alpha_p \mathbf{I}_n
	\end{bmatrix} \begin{bmatrix}
		\dot{\mathbf{p}}_{a_d}\\
		\dot{\mathbf{p}}_{r_d}
	\end{bmatrix}.
	\label{point_system}
\end{equation}

\begin{problem}[Regulation of the average position]
	Let $\mathbf{p}_{d} \in \mathbb{R}^n$ be a desired average position for the two-point system \eqref{point_system}.
	Assuming that $\dot{\mathbf{p}}_{r_d}$ is known, we want to find point velocities $\dot{\mathbf{p}}_{i_d}$ such that the error $\tilde{\mathbf{p}}_a = \mathbf{p}_d - \mathbf{p}_a$, with dynamics given by \eqref{pos_error_dyn}, asymptotically reaches $\mathbf{0}$.
	\label{two_point_problem}
\end{problem}
Note that, when $n = 3$, Problem \ref{two_point_problem} is equivalent to the translational part of Problem \ref{cooperative_problem}, i.e., the problem of regulating the absolute position of a dual-arm manipulation system.

\subsection{Compensating for the effect of asymmetries}
Let $\alpha_p$ be a constant, i.e., $f_p(\tilde{\mathbf{p}}_a, \alpha_p) = 0$, and $V(\tilde{\mathbf{p}}_a) = \frac{1}{2}\tilde{\mathbf{p}}_a^\top\tilde{\mathbf{p}}_a $ be a Lyapunov function candidate.
If we take the time derivative of $V(\tilde{\mathbf{p}}_a)$ along the trajectories of the system \eqref{pos_error_dyn} we get
\begin{equation}
	\dot{V}(\tilde{\mathbf{p}}_a ) =  -\tilde{\mathbf{p}}_a^\top \left((\alpha_p - 0.5)\dot{\mathbf{p}}_{r_d} + \dot{\mathbf{p}}_{a_d}\right),
\end{equation}
and we can design the  absolute control law as
\begin{equation}
	\dot{\mathbf{p}}_{a_d} = -(\alpha_p - 0.5)\dot{\mathbf{p}}_{r_d} + \mathbf{K}_p\tilde{\mathbf{p}}_a,
	\label{trivial}
\end{equation}
with $\mathbf{K}_p \in \mathbb{R}^{n \times n}$ being a positive definite gain matrix, yielding $\dot{V}(\tilde{\mathbf{p}}_a) = -\tilde{\mathbf{p}}_a^\top\mathbf{K}_p \tilde{\mathbf{p}}_a \leq 0$, which is negative definite.
Additionally, given that there exists a $\lambda > 0$ such that $\dot{V}(\tilde{\mathbf{p}}_a) \leq -\lambda V(\tilde{\mathbf{p}}_a)$, the exponential convergence of the absolute position error $\tilde{\mathbf{p}}_a \rightarrow \mathbf{0}$ is guaranteed.
The solution \eqref{trivial} compensates for the effects of the asymmetric resolution of $\dot{\mathbf{p}}_{r_d}$ in \eqref{pos_error_dyn}.
However, this results in a symmetric solution to the relative motion task, thus cancelling out the effect of $\alpha_p$. 
In fact, under the control law \eqref{trivial}, the point system \eqref{point_system} will have dynamics
\begin{equation}
	\begin{bmatrix}
		\dot{\mathbf{p}}_1\\
		\dot{\mathbf{p}}_2
	\end{bmatrix} = \begin{bmatrix}
		\mathbf{I}_n & -\frac{1}{2}\mathbf{I}_n\\
		\mathbf{I}_n & \frac{1}{2}\mathbf{I}_n
	\end{bmatrix} \begin{bmatrix}
		\mathbf{K}_p\tilde{\mathbf{p}}_{a}\\
		\dot{\mathbf{p}}_{r_d}
	\end{bmatrix},
\end{equation}
which corresponds to the CTS formulation for the two-point system.

\subsection{Inducing asymmetries}
The solution \eqref{trivial} is analogous to a feedback linearization of \eqref{vr_disturbance}, as we `linearize' $\dot{\mathbf{p}}_{a_d}$ by compensating the disturbance induced by the constant $\alpha_p \neq 0.5$.
We can try instead to adjust $\alpha_p$ through an update rule $\dot{\alpha}_p = f_p(\tilde{\mathbf{p}}_a, \alpha_p)$.
\begin{theorem}
	The absolute motion feedback law
	\begin{equation}
		\dot{\mathbf{p}}_{a_d} = \mathbf{K}_p\tilde{\mathbf{p}}_a,	
		\label{pa_law}
	\end{equation}
	together with the following update law for the cooperation parameter $\alpha_p$,
%	\dot{\alpha}_p = - \gamma \alpha_p (\alpha_p - 1) (\tilde{\mathbf{p}}_a^\top\dot{\mathbf{p}}_{r_d} + k_\alpha_p(\alpha_p - 0.5)),
	\begin{equation}
		f_p(\tilde{\mathbf{p}}_a, \alpha_p) = \gamma_p (\tilde{\mathbf{p}}_a^\top\dot{\mathbf{p}}_{r_d} - (\alpha_p - 0.5)(k_p + |\tilde{\mathbf{p}}_a^\top\dot{\mathbf{p}}_{r_d}|)),
		\label{alpha_law_1}
	\end{equation}
	with $\gamma_p, k_p > 0$, applied to \eqref{pos_error_dyn}, ensures the exponential stability of \eqref{pos_equi}. 	\label{pos_theo}
	\end{theorem}
	\begin{proof}
		We define the Lyapunov candidate function
		\begin{equation}
			V(\tilde{\mathbf{p}}_a, \alpha) = V_a(\tilde{\mathbf{p}}_a) + \frac{1}{\gamma_p}V_\text{asym}(\alpha_p),
			\label{lyap_1}
		\end{equation}
		such that
		\begin{align}
			V_a(\tilde{\mathbf{p}}_a) &= \frac{1}{2} \tilde{\mathbf{p}}_a^\top\tilde{\mathbf{p}}_a\label{pa_lyap}\\
			V_\text{asym}(\alpha) &= \frac{1}{2}(\alpha - 0.5)^2 \label{quadratic_alpha}.
		\end{align}
		% V_\text{asym}(\alpha) &= \frac{1}{2} \log\left(\frac{1}{1 - \left(\frac{\alpha - 0.5}{0.5}\right)^2} \right) \label{barrier}
		Under the control law \eqref{pa_law} and update rule \eqref{alpha_law_1}, the time derivative of \eqref{lyap_1} along the trajectories of \eqref{pos_error_dyn} yields $\dot{V}(\tilde{\mathbf{p}}_a, \alpha_p) = - (k_{p} + |\tilde{\mathbf{p}}_a^\top\dot{\mathbf{p}}_{r_d}|) (\alpha_p - 0.5)^2 - \tilde{\mathbf{p}}_a^\top\mathbf{K}_a \tilde{\mathbf{p}}_a \leq 0$, which is negative definite.
		Furthermore, there exists $\lambda > 0$ such that $\dot{V}(\tilde{\mathbf{p}}_a, \alpha_p) \leq -\lambda V(\tilde{\mathbf{p}}_a, \alpha_p)$ and thus the equilibrium point \eqref{pos_equi} is exponentially stable.
	\end{proof}

\subsection{Constraining the degree of cooperation}
The solution \eqref{alpha_law_1} can lead to $\alpha_p \not \in D_\alpha$ and thus the original semantics of the cooperation parameter can be lost. 
Alternatively, we can constrain $\alpha_p$ by using logarithmic barrier functions.

\begin{lemma}
	Let $\alpha_p(t_0) = 0.5$, for some arbitrary initial time $t = t_0$. Then, the update rule
	\begin{equation}
			\begin{aligned}
				f_p(\tilde{\mathbf{p}}_a, \alpha_p) = \gamma_p &\alpha_p (1 - \alpha_p) (\tilde{\mathbf{p}}_a^\top\dot{\mathbf{p}}_{r_d} \\
				&- (\alpha_p - 0.5)(k_{p} + |\tilde{\mathbf{p}}_a^\top\dot{\mathbf{p}}_{r_d}|)),
			\end{aligned}
			\label{alpha_law_2}
	\end{equation}
	together with the control law \eqref{pa_law}, ensures the asymptotic stability of \eqref{pos_equi}.
	Additionally, \eqref{alpha_law_2} guarantees $\alpha_p \in D_\alpha$ $\forall t \geq t_0$.
		\label{pos_barrier_lemma}
\end{lemma}
	\begin{proof}
		Consider the Lyapunov candidate function
		\begin{equation}
			V(\tilde{\mathbf{p}}_a, \alpha_p) = V_a(\tilde{\mathbf{p}}_a) + \frac{1}{\gamma_p}V_\text{barrier}(\alpha_p),
			\label{lyap_2}
		\end{equation}
		where $V_a(\tilde{\mathbf{p}}_a)$ is defined in \eqref{pa_lyap} and $V_\text{barrier}(\alpha_p): (0, 1) \rightarrow \mathbb{R}$ is a logarithmic barrier function,
		\begin{equation}
			V_\text{barrier}(\alpha) = \frac{1}{2} \log\left(\frac{1}{1 - \left(\frac{\alpha - 0.5}{0.5}\right)^2} \right),
			\label{barrier}
		\end{equation}
		where $\log$ denotes the natural logarithm.
		Then, the time derivative of \eqref{lyap_2} along the trajectories of \eqref{pos_error_dyn} yields $\dot{V}(\tilde{\mathbf{p}}_a, \alpha_p) = - (k_{\alpha_p} + |\tilde{\mathbf{p}}_a^\top\dot{\mathbf{p}}_{r_d}|) (\alpha_p - 0.5)^2 - \tilde{\mathbf{p}}_a^\top\mathbf{K}_a \tilde{\mathbf{p}}_a \leq 0$ which guarantees the asymptotic stability of the equilibrium point \eqref{pos_equi}.	Furthermore, the negative definiteness of $\dot{V}(\tilde{\mathbf{p}}_a, \alpha_p)$ implies $V(\tilde{\mathbf{p}}_a, \alpha_p) \leq V(\tilde{\mathbf{p}}_a(t_0), \alpha_p(t_0))$ which in turn results in $V_\text{barrier}(\alpha_p) \leq \gamma_p V(\tilde{\mathbf{p}}_a(t_0), \alpha_p(t_0))$.
		Therefore, since $\alpha_p(t_0) = 0.5$, $\alpha_p \in (0, 1) \subset D_\alpha$ is guaranteed for all time $t \geq t_0$.
	\end{proof}
\subsection{Numerical results}
To illustrate the effect of update laws \eqref{alpha_law_1} and \eqref{alpha_law_2}, we run a set of simulations on a 1-dimensional problem and for a planar setup, i.e., $n = 1$ and $n = 2$, respectively.
In both cases,  we adopt 
\begin{equation}
	\mathbf{v}_{r_d} = \mathbf{p}_1 - \mathbf{p}_2,
	\label{relative_control}
\end{equation}
that is, we assume the relative motion task for the two point system is to bring the points together.
This is analogous to tasks such as assembly, where the goal is to mate two parts, and can be represented by the problem of aligning two frames rigidly attached to the end-effector, Fig. \ref{diagram}.
Alternatively, $\mathbf{v}_{r_d}$ could be a periodic signal, e.g., as in a pruning task, or the resultant of a task-specific controller \cite{Almeida2018}.
In all simulations, we use $\mathbf{K}_p = \mathbf{I}_n$.

\begin{figure}[t]
	\centering
		\includegraphics[width = \figsize]{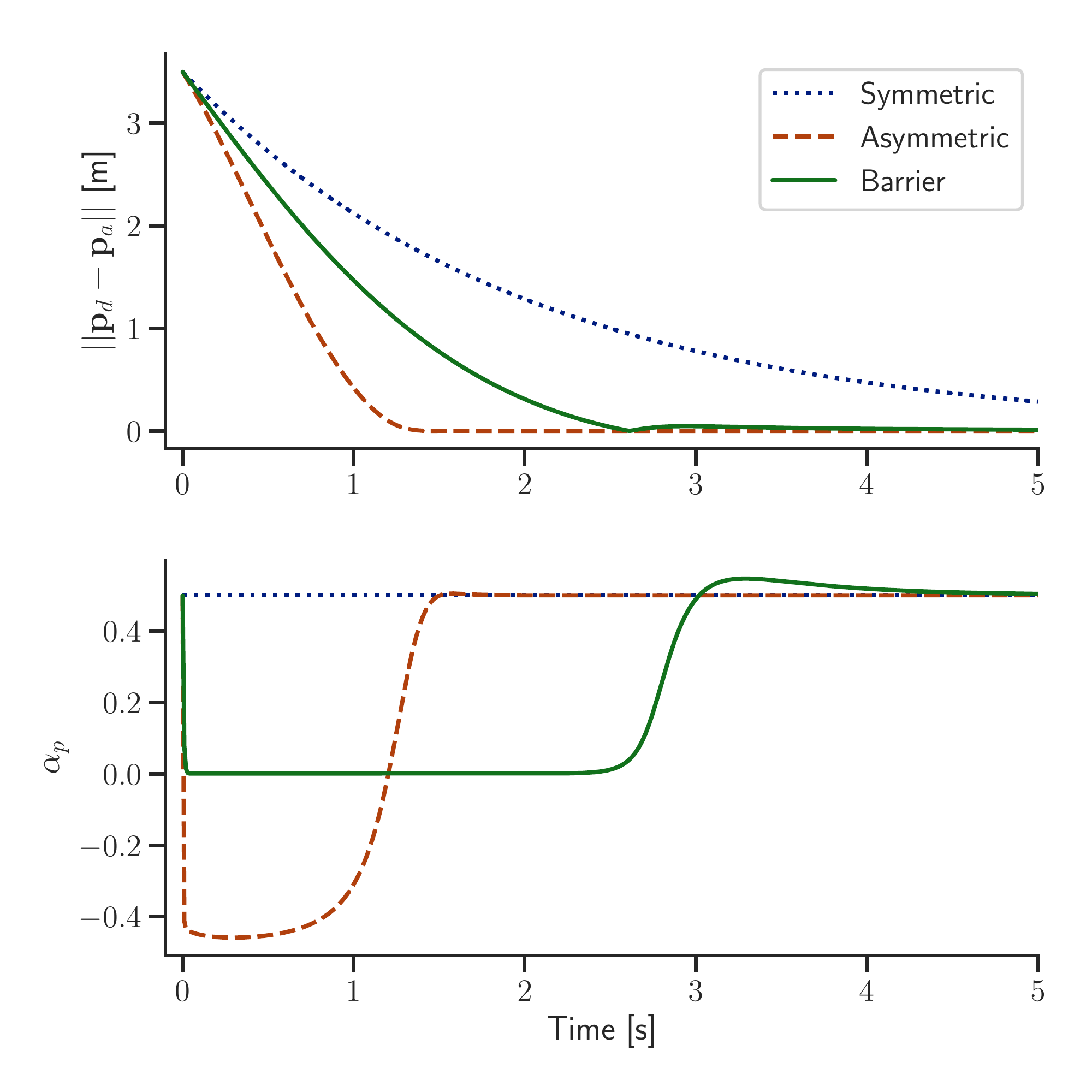}
	\caption{Effect of the induced asymmetries on the system \eqref{point_system} for a 1-dimensional scenario. By dynamically changing the cooperation of the two points on the relative motion task, we are able to accelerate the convergence of the error of the absolute motion task. \label{results_onedim}}
\end{figure}
\subsubsection{Results for 1-dimensional points}
We initialize the system at $\mathbf{p}_1(t_0) = 0$, $\mathbf{p}_2(t_0) = 1$ and $\alpha_p(t_0) = 0.5$, and depict results for $\mathbf{p}_{d} = 3$, $\gamma_p = 100$ and $k_p = 0.25$ in Fig. \ref{results_onedim}.
We labelled the results obtained while using update law \eqref{alpha_law_1} as `Asymmetric' and the results when using \eqref{alpha_law_2} as `Barrier'. 
The case for $\alpha_p = 0.5$ is denoted `Symmetric'. 
The results show how an asymmetric execution of the relative motion task can positively contribute to the convergence rate of the norm of the absolute task error, $||\tilde{\mathbf{p}}_a||$.
It additionally illustrates that, by allowing $\alpha_p \not \in D_{\alpha}$, the convergence rate of the absolute motion task can be accelerated, since the update law \eqref{alpha_law_1} enables the amplification of the effect of $\dot{\mathbf{p}}_{r_d}$ on the absolute motion space. 
Conversely, the update law \eqref{alpha_law_2} constrains $\alpha_p$ to $D_\alpha$ and sticks to the analogy of using $\alpha_p$ as a parameter which enables master-slave, symmetric and blended modes of operation.

Additionally, we depict how $k_{p}$ contributes to the dampening of the update laws for $\alpha_p$ in Fig. \ref{damping_results} for update law \eqref{alpha_law_1} and in Fig. \ref{damping_results_barrier} for the law \eqref{alpha_law_2}.
In both cases, higher $k_{p}$ leads to a more conservative rate of change of $\alpha_p$.
However, lower values may result in an overshoot of $\tilde{\mathbf{p}}_a$ and hinder the convergence rate of the absolute motion error.
In particular, when $k_p = 0$, the convergence result $\alpha_p \rightarrow 0.5$ is no longer guaranteed, as can be observed.

\subsubsection{Planar case}
The disturbance induced by $\dot{\mathbf{p}}_{r_d}$ affects the absolute position $\mathbf{p}_a$ only on the subspace along which there is relative motion, as clearly seen in eq. \eqref{pos_error_dyn}. 
Therefore, for the relative control law \eqref{relative_control}, the induced disturbance will affect $\dot{\mathbf{p}}_a$ only along the line defined by $\mathbf{p}_1 - \mathbf{p}_2$. 
In higher dimensional cases, the update laws \eqref{alpha_law_1} and \eqref{alpha_law_2} will contribute to the alignment of the two points such that the projection of the target $\mathbf{p}_d$ on $\mathbf{p}_1 - \mathbf{p}_2$ is equidistant from both points.
Simulation results are depicted in Fig. \ref{plane_asym}.
In the simulation, the initial conditions were set as $\mathbf{p}_1(t_0) = [0, 0]^\top$, $\mathbf{p}_2(t_0) = [1, 0]^\top$ and $\alpha_p(t_0) = 0.5$, with target average position $\mathbf{p}_d = [-0.1, 0.1]^\top$ and $\gamma_p = 100, k_{p} = 0.25$.
In addition to the labels used in Fig. \ref{results_onedim}, we use `Constant' for the simulated the result of adopting a master-slave approach for the relative motion task by setting a constant $\alpha_p = 1$.
Note that for the simulated initial conditions, setting $\mathbf{p}_1$ as the master in the execution of the relative motion task significantly affects the convergence of the absolute position error.

\begin{figure}[t]
	\centering
	\includegraphics[width = \figsize]{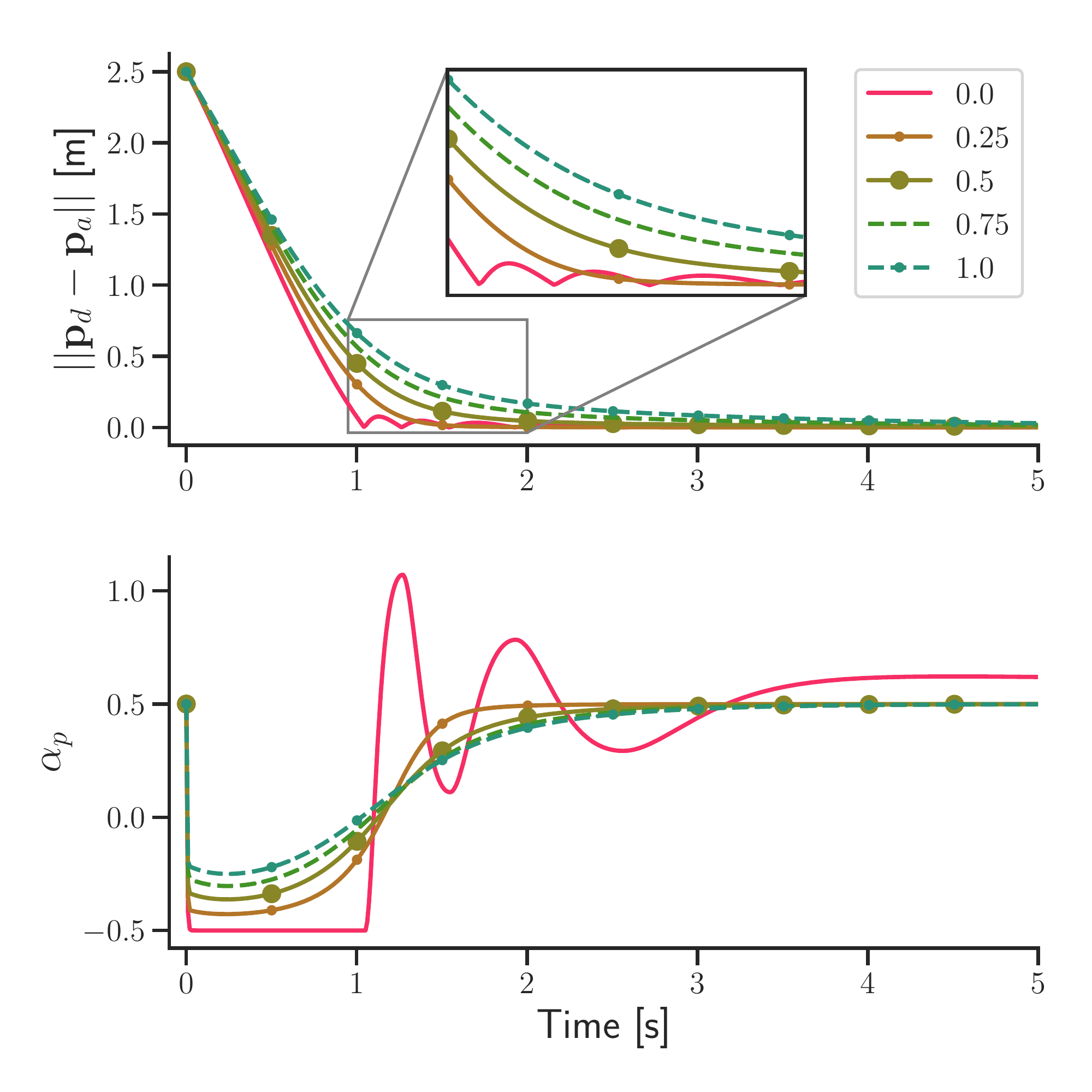}
	\caption{The effects of changing the parameter $k_{p}$ in the update law \eqref{alpha_law_1}.\label{damping_results}}
\end{figure}

\section{Dual-Arm manipulation \label{dual_arm}}
The update laws \eqref{alpha_law_1} and \eqref{alpha_law_2}, together with \eqref{pa_law}, ensure the stability of the linear component of the absolute motion task. 
To fully address Problem \ref{cooperative_problem}, we require an update law for $\alpha_\omega$ and a corresponding angular control law.
The dynamics of the absolute orientation are given in \eqref{abs_ori_prop}.
The angular velocities of the cooperative system are related to the angular part of \eqref{two_alphas_res},
\begin{equation}
	\begin{bmatrix}
		\boldsymbol{\omega}_{1_d}\\
		\boldsymbol{\omega}_{2_d}
	\end{bmatrix} = \begin{bmatrix}
		\mathbf{I}_3 & -(1 - \alpha_\omega)\mathbf{I}_3\\
		\mathbf{I}_3 & \alpha_\omega\mathbf{I}_3
	\end{bmatrix} \begin{bmatrix}
		\boldsymbol{\omega}_{a_d}\\
		\boldsymbol{\omega}_{r_d}
	\end{bmatrix},
\end{equation}
and the system's absolute angular velocity is disturbed by $\boldsymbol{\omega}_{r_d}$ analogously to \eqref{vr_disturbance}, as seen in \eqref{angular_vel}.

\begin{figure}[t]
	\centering
	\includegraphics[width = \figsize]{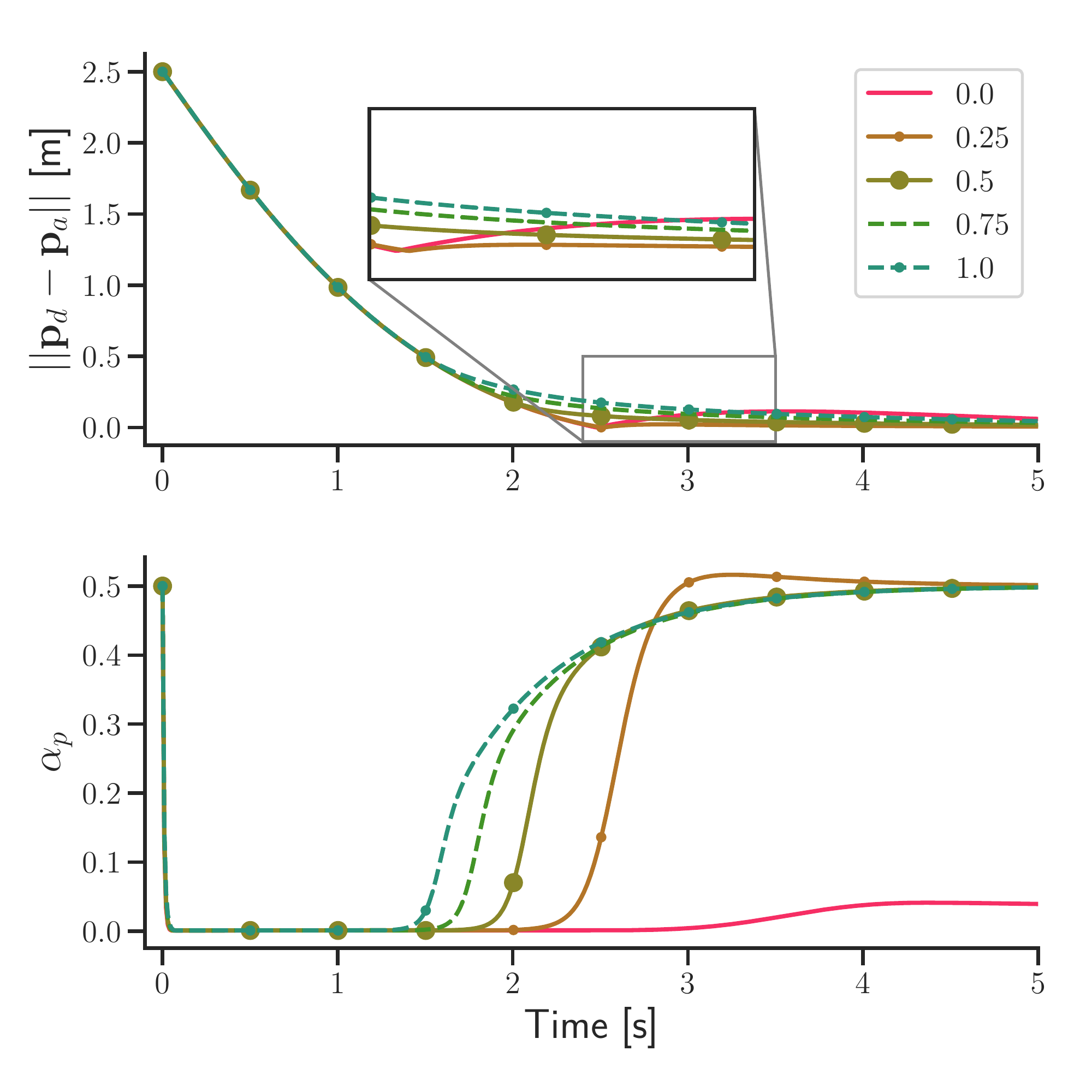}
		\caption{The effects of changing the parameter $k_{p}$ when using the update law \eqref{alpha_law_2}.\label{damping_results_barrier}}
\end{figure}
\begin{figure}[t]
	\centering
	\includegraphics[width = 0.4\textwidth]{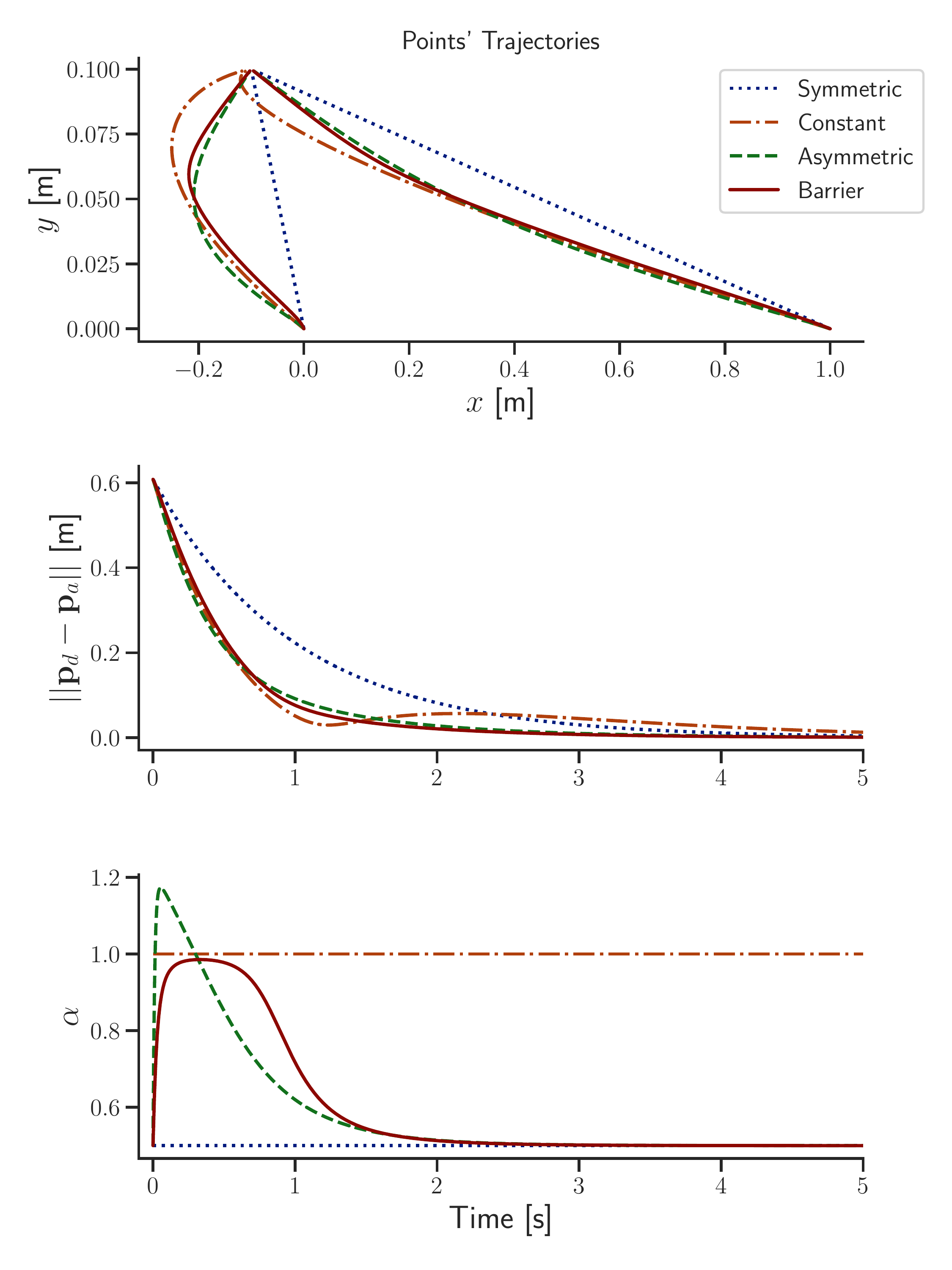}
	\caption{Adjusting the asymmetries of the relative task execution can improve the convergence rate of the absolute motion task. However, the asymmetric execution only affects the components of $\tilde{\mathbf{p}}_a$ along the direction of $\dot{\mathbf{p}}_{r_d}$. \label{plane_asym}}
\end{figure}
\begin{theorem}
	Consider the feedback control law
	\begin{equation}
		\boldsymbol{\omega}_{a_d} = \mathbf{K}_{\omega} \tilde{\boldsymbol{\epsilon}}_a
		\label{angular_control}
	\end{equation}
	and the following update equation for $\alpha_\omega$, with $\gamma_\omega, k_{\omega} > 0$,
	% \dot{\alpha}_\omega = -\gamma \alpha_\omega(\alpha_\omega - 1)(- \tilde{\boldsymbol{\epsilon}}_a^\top\boldsymbol{\omega}_{r_d} + k_\alpha_\omega(\alpha_\omega - 0.5))
	\begin{equation}
		f_\omega(\tilde{\mathcal{Q}}_a, \alpha_\omega) = \gamma_\omega (\tilde{\boldsymbol{\epsilon}}_a^\top\boldsymbol{\omega}_{r_d} - k_{\omega}(\alpha_\omega - 0.5)).
		\label{alpha_law_3}
	\end{equation}
	Together, \eqref{angular_control} and \eqref{alpha_law_3} ensure the asymptotic stability of \eqref{ang_equi}.
		\label{ori_theo}
\end{theorem}
	\begin{proof}
		Consider the following Lyapunov candidate function,
		\begin{equation}
				V(\tilde{\mathcal{Q}}_a, \alpha_\omega) =  V_\omega(\tilde{\mathcal{Q}}_a) + \frac{1}{\gamma_\omega}V_\text{asym}(\alpha_\omega),
			\label{V_ang}
		\end{equation}
		where
		\begin{equation}
			V_\omega(\tilde{\mathcal{Q}}_a) = (\eta_d - \eta_a)^2 + (\boldsymbol{\epsilon}_d - \boldsymbol{\epsilon}_a)^\top(\boldsymbol{\epsilon}_d - \boldsymbol{\epsilon}_a)
			\label{lyap_quat}
		\end{equation}
		and $V_\text{asym}(\alpha_\omega)$ is given by \eqref{quadratic_alpha}.
		The time derivative of \eqref{V_ang} along the system trajectories is given by $\dot{V}(\tilde{\mathbf{Q}}_a, \alpha_\omega) =  -k_\omega (\alpha_\omega - 0.5)^2 - \tilde{\boldsymbol{\epsilon}}_a^\top\mathbf{K}_\omega\tilde{\boldsymbol{\epsilon}}_a \leq 0$, which is negative definite and thus the equilibrium point \eqref{ang_equi} is asymptotically stable. 
	\end{proof}
As in the solution to Problem \ref{two_point_problem}, we can constrain $\alpha_\omega$ such that $\alpha_\omega \in D_\alpha$ by making use of a logarithmic barrier function.
\begin{lemma}
	Let $\alpha_\omega(t_0) = 0.5$. Then, the update rule
	\begin{equation}
			f_\omega(\tilde{\mathcal{Q}}_a, \alpha_\omega) = \gamma_\omega \alpha_\omega (1 - \alpha_\omega) (\tilde{\boldsymbol{\epsilon}}_a^\top\boldsymbol{\omega}_{r_d} - k_{\omega}(\alpha_\omega - 0.5)),
			\label{alpha_law_4}
	\end{equation}
	together with the control law \eqref{angular_control}, ensures the asymptotic stability of \eqref{ang_equi}.
	Additionally, \eqref{alpha_law_4} guarantees $\alpha_\omega \in D_\alpha$ $\forall t \geq t_0$.
		\label{ori_barrier_lemma}
\end{lemma}
	\begin{proof}
		Consider the Lyapunov candidate function
		\begin{equation}
			V(\tilde{\mathcal{Q}}_a, \alpha_\omega) = V_\omega(\tilde{\mathcal{Q}}_a) + \frac{1}{\gamma_\omega}V_\text{barrier}(\alpha_\omega),
			\label{lyap_4}
		\end{equation}
		where $V_\omega(\tilde{\mathcal{Q}}_a)$ is defined in \eqref{lyap_quat} and $V_\text{barrier}(\alpha_\omega)$ is the barrier function defined in \eqref{barrier}.
		Then, the time derivative of \eqref{lyap_4} over the trajectories of \eqref{abs_ori_prop} yields $\dot{V}(\tilde{\mathcal{Q}}_a, \alpha_\omega) = - k_{\omega} (\alpha_\omega - 0.5)^2 -  \tilde{\boldsymbol{\epsilon}
		}_a^\top\mathbf{K}_\omega\tilde{\boldsymbol{\epsilon}}_a \leq 0$.
		Additionally, the proof that $\alpha_\omega \in D_\alpha$ for all $t \geq t_0$ follows the proof that $\alpha_p \in D_\alpha$ from Lemma~\ref{pos_barrier_lemma}.
	\end{proof}
We can now state our solution to Problem \ref{cooperative_problem} as a combination of the solution to Problem \ref{two_point_problem} and the previous analysis for the control of the absolute orientation.
\begin{figure}[t!]
	\centering
	\includegraphics[width = \figsize]{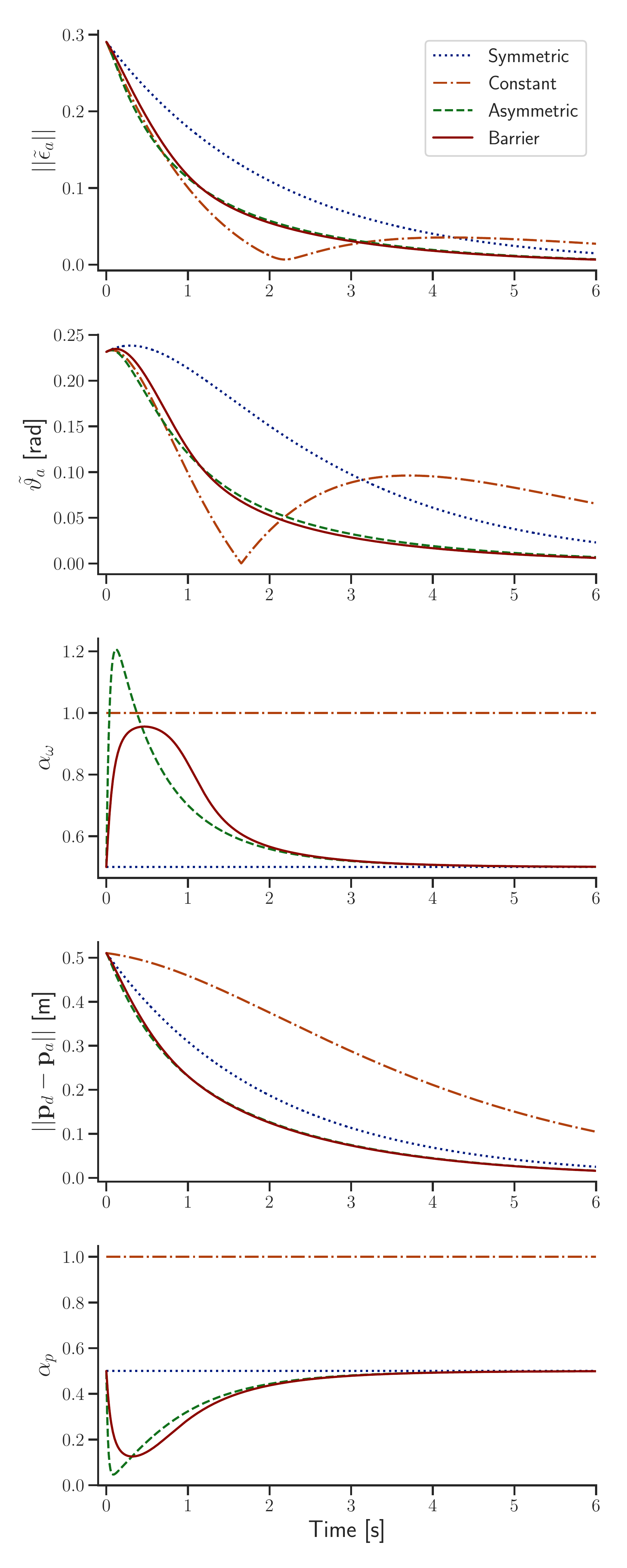}
	\caption{Results for a dual-arm cooperative task. The convergence of the absolute motion task is affected by the type of cooperation in the relative motion task.\label{results_pose}}
\end{figure}
\begin{theorem}
	Consider the linear control law \eqref{pa_law} for $n = 3$ and the angular control law \eqref{angular_control}.
	Together with the update law for $\alpha_p$ \eqref{alpha_law_1} and the update law \eqref{alpha_law_3} for $\alpha_\omega$, the cooperative manipulation system \eqref{pos_error_dyn}-\eqref{abs_ori_prop} will have an asymptotically stable equilibrium point \eqref{pos_equi}-\eqref{ang_equi}.
	If it is desired for $\alpha_p \in D_\alpha$ or $\alpha_\omega \in D_\alpha$, the update laws \eqref{alpha_law_2}  and \eqref{alpha_law_4}, respectively, can be used.
	\end{theorem}
	\begin{proof}
		The stability of the system \eqref{pos_error_dyn}-\eqref{abs_ori_prop} can be asserted by considering the Lyapunov candidate function
		\begin{equation}
			\begin{aligned}
				V(\tilde{\mathbf{p}}_a, \tilde{\mathcal{Q}}_a, \alpha_p, \alpha_\omega) &= V_a(\tilde{\mathbf{p}}_a) + V_\omega(\tilde{\mathcal{Q}}_a)\\
				 &+ \frac{1}{\gamma_p}V_\text{asym}(\alpha_p) + \frac{1}{\gamma_\omega}V_\text{asym}(\alpha_\omega),
			\end{aligned}
			\label{final_lyap}
		\end{equation}
		with definitions from Theorems \ref{pos_theo} and \ref{ori_theo}.
		The time derivative of \eqref{final_lyap} along the system's trajectories is negative everywhere except at \eqref{pos_equi}-\eqref{ang_equi}, which can be straightforwardly observed from the proofs of the aforementioned Theorems. 
		Alternatively, $V_\text{barrier}(\alpha_p)$ and $V_\text{barrier}(\alpha_\omega)$ can be used to constrain $\alpha_p$ and $\alpha_\omega$, respectively, with the stability proof following from Lemmas \ref{pos_barrier_lemma} and \ref{ori_barrier_lemma}. 
	\end{proof}

\subsection{Numerical results}
We simulate a cooperative manipulation problem where $\{h_1\}$ and $\{h_2\}$ represent two frames rigidly attached to a dual-arm robotic manipulator.
In our simulations, we ignore the effects that kinematic limitations would have on the range of possible motions of the cooperative system, and focus our analysis exclusively on the regulation of the task-space quantities.
Additionally, without loss of generality, we assume that $\mathbf{v}_r$ is generated by the feedback control law
\begin{equation}
	\mathbf{v}_r = \mathbf{K}_r \begin{bmatrix}
		\mathbf{p}_1 - \mathbf{p}_2\\
		\tilde{\boldsymbol{\epsilon}}_r
	\end{bmatrix},
\end{equation}
with $\mathbf{K}_r \in \mathbb{R}^{6\times 6}$ being positive definite and $\tilde{\boldsymbol{\epsilon}}_r$ is the vector part of the error quaternion $\tilde{\mathcal{Q}}_r = \mathcal{Q}_1 * \mathcal{Q}_2^{-1}$.
In the presented results, we set the gains of the all the update laws as $\gamma_p = \gamma_\omega = 100$ and $k_p = k_\omega = 0.25$.
We initialize the simulations with $\mathbf{p}_1(t_0) = [1, 0, -0.2]^\top$, $\mathbf{p}_2(t_0) = [0.0, 0.2, 0.25]^\top$, $\mathcal{Q}_1(t_0) = \{0.9, [0.261, 0, 0.348]^\top\}$, $\mathcal{Q}_2(t_0) = \{0.9, [-0.195, -0.389, 0]^\top\}$ and $\alpha_p(t_0) = \alpha_\omega(t_0) = 0.5$.
The target absolute pose is set as $\mathbf{p}_d = \mathbf{0}$ and $\mathcal{Q}_d = \{0.921, [0.275, 0, 0.275]^\top\}$.
The control gains are set as $\mathbf{K}_r = \mathbf{I}_6$, $\mathbf{K}_\omega = \mathbf{I}_3$ and $\mathbf{K}_p = \mathbf{I}_3$.

The numerical results are depicted in Fig. \ref{results_pose}, where the labelled plots have a similar meaning as in Fig. \ref{plane_asym}, i.e., 'Symmetric' denotes the solution computed with $\alpha_p = \alpha_\omega = 0.5$, 'Constant' adopts a master-slave resolution of the relative motion task, 'Asymmetric' depicts the solution with update laws \eqref{alpha_law_1} and \eqref{alpha_law_3} and finally 'Barrier' uses the update laws \eqref{alpha_law_2} and \eqref{alpha_law_4}.
We plot the norm of the vector part of $\tilde{\mathcal{Q}}_a$, $||\tilde{\boldsymbol{\epsilon}}_a||$, as well as the angle of the angle-axis representation of $\mathbf{R}_a$, $\tilde{\vartheta}_a$ to illustrate the evolution of the absolute orientation of the cooperative system.
The simulation reinforces the argument that inducing deliberate asymmetries on the execution of a relative motion task can benefit the convergence rate of the absolute motion task in a cooperative manipulation system, while that arbitrary asymmetries might significantly disturb the execution of the absolute task.

\section{Conclusions}
In this article, we consider the problem of the cooperative control of robotic manipulators.
We investigate how, in the context of a CTS task specification, the asymmetric resolution of the relative motion task affects the execution of the absolute motion of the system.
The observation that, in the presence of asymmetries, the relative motion acts as a disturbance to the absolute task leads to the proposal of update laws to the degree of cooperation of the two arms in the relative motion task.
This is a novel approach to the design of cooperative control laws of cooperative manipulation systems.
We show, through numerical simulations, that by deliberately changing the mode of cooperation of the arms, the convergence rate of the absolute motion task can be improved.
The proposed methods can be straightforwardely implemented in a robotic manipulator through, e.g., the employment of the ECTS Jacobian \cite{Park2016} within a suitable kinematic control framework \cite{Kanoun2011, Escande2014}, to account for robot kinematic limitations.

%\clearpage
\bibliographystyle{unsrt} 
\bibliography{biblio.bib}
\end{document}